\documentclass[11pt,journal,onecolumn]{IEEEtran}
\usepackage{pdfpages}
\usepackage{setspace}
\usepackage[utf8]{inputenc}
\usepackage[T1]{fontenc}
\usepackage{amsmath,amssymb,amsfonts,latexsym,amsthm,enumerate,fullpage,hyperref,dsfont}
\usepackage{stackrel,color,mathtools}
\usepackage[noadjust]{cite}
\usepackage{graphicx}
\usepackage{slashbox}
\usepackage{setspace}

\mathtoolsset{showonlyrefs=true}
 
\newtheorem{theorem}{Theorem}
\newtheorem{lemma}{Lemma}
\newtheorem{definition}{Definition}

\newtheorem{example}{Example} 
\newcommand{\dfn}{\triangleq}
\newcommand{\bs}[1]{\boldsymbol{#1}}

\newcommand{\indfunc}[1]{\mathds{1}\left(#1\right)}

\def\tind{i}
\newcommand{\blockind}{j}
\def\protLength{n}
\def\blockLength{m}
\def\smallwindow{p}

\newcommand{\FiniteStateProtocol}{\Phi}
\newcommand{\StateOrder}{M}

\def\BSCprob{\varepsilon}

\newcommand{\sigmalength}{N}
\newcommand{\transFunc}{\phi}
\newcommand{\entireprot}{\pi}
\newcommand{\protTrans}{\tau}
\newcommand{\statevar}{s}
\newcommand{\statespace}{S}
\newcommand{\finitestatefunc}{\psi}
\newcommand{\stateadvancefunc}{\eta}

\newcommand{\Cshannon}{{\mathsf{C}}_{\mathsf{Sh}}}
\newcommand{\Cinter}{{\mathsf{C}}_{\mathsf{I}}}
\newcommand{\simulatingprotocol}{\Sigma}
\newcommand{\coincidingConstant}{K}
\newcommand{\usefulset}{\mathcal{F}}
\newcommand{\Pe}{P_e}
\newcommand{\speakerorderfunc}{\mu}
\newcommand{\stateiterfunc}{\nu}
\DeclareMathOperator*{\sumparity}{{\oplus}}

\begin{document}
\title{On the Interactive Capacity of Finite-State Protocols}
\author{Assaf Ben-Yishai, Young-Han Kim, Rotem Oshman, and Ofer Shayevitz
	\thanks{A. Ben-Yishai and O. Shayevitz are with the Department of EE--Systems, Tel Aviv University, Tel Aviv, Israel. 
		Y.-H.~Kim is with the Department of Electrical and Computer Engineering, University of California, San Diego, La Jolla, CA 92093 USA. 
		R. Oshman is with the 
		Department of Computer Science, Tel Aviv University, Tel Aviv, Israel. 
		Emails: \{assafbster@gmail.com, yhk@ucsd.edu, roshman@tauex.tau.ac.il,  ofersha@eng.tau.ac.il\}
		The work of A. Ben-Yishai was supported by an ISF grant no. 1367/14. 
		The work of O. Shayevitz was supported by an ERC grant no. 639573.
	This paper was presented in part at ISIT 2019.}}

\maketitle

\begin{abstract}
The interactive capacity of a noisy channel is the highest possible rate at which arbitrary interactive protocols can be simulated reliably over the channel. Determining the interactive capacity is notoriously difficult, and the best known lower bounds are far below the associated Shannon capacity, which serves as a trivial (and also generally the best known) upper bound. This paper considers the more restricted setup of simulating finite-state protocols. It is shown that all two-state protocols, as well as rich families of arbitrary finite-state protocols, can be simulated at the Shannon capacity, establishing the interactive capacity for those families of protocols.
\end{abstract}


\section{Introduction} 
In the classical one-way communication problem, a transmitter (Alice) wishes to reliably send a message to a receiver (Bob) over a memoryless noisy channel. She does so by mapping her message into a sequence of channel inputs (a \textit{codeword}) in a predetermined way, which is corrupted by the channel and then observed by Bob, who tries to recover the original message. The {\em Shannon capacity} of the channel, 
quantifies the most efficient way of conducting reliable communication, and is defined as the maximal number of message bits per channel-use that Alice can convey to Bob with vanishingly low error probability. In the two-way channel setup \cite{shannon1961two}, both parties draw independent messages and wish to exchange them over a two-input two-output memoryless noisy channel, and the Shannon capacity (region) is defined similarly. Unlike the one-way case, both parties can now employ adaptive coding, by incorporating their respective observations of the past channel outputs into their transmission processes. However, just as in the one-way setup, the messages they wish to exchange are determined before communication begins. In other words, if Alice and Bob had been connected by a noiseless bit pipe, they could have simply sent their messages without any regard to the message of their counterpart.  

In a different two-way communication setup, generally referred to as {\em interactive communication}, the latter assumption is no longer held true. In this interactive communication setup, Alice and Bob do not necessarily wish to disclose all their local information. What they want to tell each other depends, just like in human conversation, on what the other would tell them. A simple instructive example (taken from \cite{gelles2015coding}) is the following. Suppose that Alice and Bob play \textit{correspondence chess}. Namely, they are located in two distinct places and play by announcing their moves over a communication channel (using, say, $12$ bits per move, which is clearly sufficient). If the moves are conveyed without error, then both parties can keep track of the state of the board, and the game can proceed to its termination. The sequence of moves occurring over the course of this noiseless game is called a {\em transcript}, and it is dictated by the {\em protocol} of the game, which constitutes Alice and Bob's respective strategies determining their moves at any given state of the board.  

Now, assume that Alice and Bob play a chess game over a noisy two-way channel, yet wish to simulate the transcript as if no noise were present. In other words, they would like to communicate back and forth in a way that ensures, once communication is over, that the transcript of the noiseless game can be reproduced by to both parties with a small error probability. They would also like to achieve this goal as efficiently as possible, i.e., with the least number of channel uses. One direct way to achieve this is by having both parties describe their entire protocol to their counterpart, i.e., each and every move they might take given each and every possible state of the board. This reduces the interactive problem to a non-interactive one, with the protocol becoming a pair of messages to be exchanged. However, this solution is grossly inefficient; the parties now know much more than they really need to in order to simply reconstruct the transcript. At the other extreme, Alice and Bob may choose to describe the transcript itself by encoding each move separately on-the-fly, using a short error correcting code. Unfortunately, this code must have some fixed error probability, and hence an undetected error is bound to occur at some unknown point, causing the states of the board held by the two parties to diverge, and rendering the remainder of the game useless. It is important to note that if Alice and Bob had wanted to play sufficiently many games in parallel, then they could have used a long error-correcting code to simultaneously protect the set of all moves taken at each time point, which in principle would have let them operate at the one-way Shannon capacity (which is the best possible). The crux of the matter therefore lies in the fact that the interactive problem is \textit{one-shot}, namely, only a single instance of the game is played.

In light of the above, it is perhaps surprising that it is nevertheless possible to simulate any one-shot interactive protocol using a number of channel uses that is proportional to the length of the transcript, or in other words, that there is a positive {\em interactive capacity} whenever the Shannon capacity is positive. This fact was originally proved by Schulman \cite{schulman1992communication}, who was also the first to introduce the notion of interactive communication over noisy channels. The lower bound on the interactive capacity was recently studied in \cite{InteractiveLowerBound}, and was found to be at least $0.0302$ of the Shannon capacity for all binary memoryless symmetric channels.
 
In this work, rather than giving a lower bound on the interactive capacity for \textit{any} protocol, we study the notion of interactive capacity where constraints are imposed on the family of protocols to be simulated.
We define the family of \textit{finite-state} protocols, and show that for a large class of these protocols, the Shannon capacity is achievable. In particular, we prove that Shannon capacity is achievable for all protocols having only two states (\textit{two-state protocols}). For larger state-spaces, 
we discuss rich families of protocols which satisfy two sufficient conditions, and show that  within these families almost all members can be reliably simulated at the Shannon capacity. We note that the approach of studying the interactive capacity of protocols having a specific structure was previously taken in \cite{haeupler2017bridging}. The authors of \cite{haeupler2017bridging} limited the ``interactiveness" of the protocols by considering families of protocols whose transcript is predictable to a certain extent, and proved that they can be simulated in higher rates than general protocols. The constraints imposed on the protocols in this paper are, however, on the memory of the protocols and not on their predictability.

The rest of the paper is organized as follows: in Section~\ref{section:formulation} the interactive communication problem is  formulated. In Section~\ref{section:finitestateformulation}, finite-state (or \textit{$\StateOrder$-state}) protocols, which are the main model discussed in this paper, are defined. In Section~\ref{section:codingscheme} the basic concepts of the coding schemes are presented. In Section~\ref{section:twostates} a capacity achieving coding scheme for two-state protocols is presented. 
In Section~\ref{section:threestates} it is proved that 
the concepts in Section~\ref{section:codingscheme} cannot be used for at least one three-state protocol. In Section~\ref{section:manystates} families of finite-state protocols for which almost-all members can be simulated at Shannon capacity are presented. Finally, Section~\ref{section:conclusion} concludes the paper.

A preliminary version of some of the results in this paper appeared in \cite{MarkovianISIT}. Here we extend upon the results of \cite{MarkovianISIT} as follows: first, \cite{MarkovianISIT} only considered \emph{Markovian protocols}, a special case of the type of protocols we consider here. Second, \cite{MarkovianISIT} gave a simple special case of the coding scheme of Section~\ref{section:codingscheme}; here we generalize the scheme and give two methods that can handle more complex protocols, beyond Markovian. Finally, the Shannon capacity achieving scheme for two-state protocols in Section~\ref{section:twostates}, the inachievability results for three states in Section~\ref{section:threestates} and the scheme for higher order models in Section~\ref{section:manystates}
appear here for the first time.


\section{The Interactive Communication Problem \label{section:formulation}}
In this paper, we define a \textit{length-$n$ interactive protocol} as a triplet 
$\bs{\entireprot}\dfn(\bs{\transFunc}^\mathrm{Alice},\bs{\transFunc}^\mathrm{Bob},\bs{\speakerorderfunc})$, where:
\begin{align}
\bs{\transFunc}^\mathrm{Alice} &\dfn \left\{{\transFunc}^\mathrm{Alice}_{\tind}:\{0,1\}^{i-1}\mapsto \{0,1\}\right\}_{\tind=1}^\protLength\\
\bs{\transFunc}^\mathrm{Bob} &\dfn \left\{{\transFunc}^\mathrm{Bob}_{\tind}:\{0,1\}^{i-1}\mapsto \{0,1\}\right\}_{\tind=1}^\protLength\\ 
\bs{\speakerorderfunc} &\dfn \left\{{\speakerorderfunc}_{\tind}:\{0,1\}^{i-1}\mapsto \{\mathrm{Alice},\mathrm{Bob}\}\right\}_{\tind=1}^\protLength.
\end{align}

The functions $\bs{\transFunc}^\mathrm{Alice}$ are known only to Alice, and the functions $\bs{\transFunc}^\mathrm{Bob}$ are known only to Bob. The \textit{speaker order functions} $\bs{\speakerorderfunc}$ are known to both parties. The \textit{transcript} $\bs{\protTrans}$ associated with the protocol $\bs{\entireprot}$ is sequentially generated by Alice and Bob as follows:
\begin{align}\label{eq:transfunc}
\protTrans_{\tind} &=\begin{cases}
{\transFunc}^\mathrm{Alice}_{\tind}(\bs{\protTrans}^{\tind-1}) & \sigma_\tind=\mathrm{Alice}\\
{\transFunc}^\mathrm{Bob}_{\tind}(\bs{\protTrans}^{\tind-1}) & \sigma_\tind=\mathrm{Bob},
\end{cases}
\end{align}
where $\sigma_\tind$ is the identity of the speaker at time $\tind$, which is given by:
\begin{align}\label{eq:transfunc2}
\sigma_\tind&={\speakerorderfunc}_{\tind}(\bs{\protTrans}^{\tind-1}).
\end{align}
In the \textit{interactive simulation problem}, Alice and Bob would like to \textit{simulate} the 
transcript $\bs{\protTrans}$, by communicating back and forth over a noisy memoryless channel $P_{Y|X}$. 
Specifically, we restrict our discussion to channels with a binary input alphabet $\mathcal{X} = \{0,1\}$, and a general (possibly continuous) output alphabet $\mathcal{Y}$. We use $\Cshannon(P_{Y|X})$ to denote the Shannon capacity of the channel. Note that $\Cshannon(P_{Y|X})\leq 1$, since the input of the channel is binary. Naturally, we also limit the discussion to channels whose Shannon capacity is non-zero.

Note that while the order of speakers in the interactive protocol itself might be determined on-the-fly (by the sequence of functions $\bs{\speakerorderfunc}$), we restrict the simulating protocol to use a predetermined order of speakers. The reason is, that allowing an adaptive order over a noisy channel, will lead to a non-zero probability of disagreement regarding the order of speakers. This disagreement might lead to simultaneous transmissions at both parties, which is not supported by the chosen physical channel model

To achieve their goal, Alice and Bob employ a length-$\sigmalength$ coding scheme $\simulatingprotocol$ that uses the channel $\sigmalength$ times. The coding scheme consists of a disjoint partition $\tilde{A}\sqcup\tilde{B} = \{1,...,\sigmalength\}$ where $\tilde{A}$ (resp. $\tilde{B}$) is the set of time indices where Alice (resp. Bob) speaks. This disjoint partition can be a function of $\bs{\speakerorderfunc}$, but not of $\bs{\transFunc}^\mathrm{Alice},\bs{\transFunc}^\mathrm{Bob}$. At time $j\in \tilde{A}$ (resp. $j\in \tilde{B}$), Alice (resp. Bob) sends some Boolean function of $(\bs{\transFunc}^\mathrm{Alice}, \bs{\speakerorderfunc}$) (resp. $(\bs{\transFunc}^\mathrm{Bob}, \bs{\speakerorderfunc}$)), and of everything she has received so far from her counterpart. The rate of the scheme is $R = \frac{\protLength}{\sigmalength}$ bits per channel use. When communication terminates, Alice and Bob produce their \textit{simulations} of the transcript $\protTrans$, denoted  by $\hat{\bs{\protTrans}}_A(\simulatingprotocol,\bs{\transFunc}^\mathrm{Alice}, \bs{\speakerorderfunc})\in\{0,1\}^\protLength$ and  $\hat{\bs{\protTrans}}_B(\simulatingprotocol,\bs{\transFunc}^\mathrm{Bob}, \bs{\speakerorderfunc})\in\{0,1\}^\protLength$ respectively. The error probability attained by the coding scheme is the probability that either of these simulations is incorrect, i.e., 
\begin{align}\label{eq:proterror}
\Pe(\simulatingprotocol,\bs{\entireprot})\dfn \Pr\left(
\hat{\bs{\protTrans}}_A(\simulatingprotocol,\bs{\transFunc}^\mathrm{Alice}, \bs{\speakerorderfunc})\neq \bs{\protTrans} \;\vee\; \hat{\bs{\protTrans}}_B(\simulatingprotocol,\bs{\transFunc}^\mathrm{Bob}, \bs{\speakerorderfunc})
\neq \bs{\protTrans}\right).
\end{align}
A rate $R$ is called \textit{achievable} if there exists a sequence $\simulatingprotocol_{\protLength}$ of length-$N_n$ coding schemes with rates $\frac{\protLength}{\sigmalength_{\protLength}}\geq R$, such that 
\begin{align}\label{eq:rateerror}
\lim_{\protLength\to \infty}\max_{\bs{\entireprot} \text{ of length } \protLength} P_e(\simulatingprotocol_{\protLength},\bs{\entireprot})  = 0, 
\end{align}
where the maximum is taken over all length-$\protLength$ interactive protocols. Accordingly, we define the interactive capacity  $\Cinter(P_{Y|X})$ as the maximum of all achievable rates for the channel $P_{Y|X}$. Note that this definition parallels the definition of maximal error capacity in the one-way setting, as we require the error probability attained by the sequence of coding schemes to be upper bounded by a vanishing term \textit{{uniformly for all protocols}}.

It is clear that at least $\protLength$ bits need to be exchanged in order to reliably simulate a general protocol, and hence the interactive capacity satisfies $\Cinter(P_{Y|X})\leq 1$. In the special case of a noiseless channel, i.e., where the output deterministically reveals the input bit, and assuming that the order of speakers is predetermined (namely $\bs{\speakerorderfunc}$ contains only constant functions), this upper bound can be trivially achieved; Alice and Bob can simply evaluate and send $\protTrans_{\tind}$ sequentially according to \eqref{eq:transfunc} and \eqref{eq:transfunc2}. Note, however, that if the order of speakers is general, then this is not a valid solution, since we required the order of speakers in the coding scheme to be fixed in advance. Nevertheless, any general interactive protocol can be sequentially simulated using the channel $2\protLength$ times with alternating order of speakers, where each party sends a dummy bit whenever it is not their time to speak. Conversely, a factor two blow-up in the protocol length in order to account for a non predetermined order of speakers is also necessary. To see this, consider an example of a protocol where Alice's first bit determines the identity of the speaker for the rest of time; in order to simulate this protocol using a predetermined order of speakers, it is easy to see that at least $n-1$ channel uses must be allocated to each party in advance. We conclude that under our restrictive capacity definition, the interactive capacity of a noiseless channel is exactly $\frac{1}{2}$.

When the channel is noisy, a tighter trivial upper bound holds:
\begin{align}\label{eq:shannon_bound}
\Cinter(P_{Y|X})\leq\frac{1}{2}\Cshannon(P_{Y|X}), 
\end{align}
To see this, consider the same example given above, and note that each party must have sufficient time to reliably send $n-1$ bits over the noisy channel. Hence, the problem reduces to a pair of one-way communication problems, in which the Shannon capacity is the fundamental limit. We remark that it is reasonable to expect the bound~\eqref{eq:shannon_bound} to be loose, since general interactive protocols cannot be trivially reduced to one-way communication as the parties cannot generate their part of the transcript without any interaction. However, the tightness of the bound remains a wide open question. 

In the remainder of this paper we limit the discussion to protocols in which the order of speakers is predetermined and {bit vs.~bit}. Namely, Alice speaks at odd times ($\sigma_\tind=\mathrm{Alice}$ for odd $\tind$) and Bob speaks at even times ($\sigma_\tind=\mathrm{Bob}$ for even $\tind$). We note that for such protocols, the $1/2$ penalty required for the adaptive order of speakers in not needed and the upper bound is therefore 
\begin{align}
\Cinter(P_{Y|X})\leq\Cshannon(P_{Y|X}). 
\end{align}

\subsection{Background \label{subsection:previous}}
The interactive communication problem introduced by Schulman \cite{schulman1992communication, schulman1996coding} is motivated by Yao's communication complexity paradigm \cite{yao1979some}. In this paradigm, the input of a function $f$ is distributed between Alice and Bob, who wish to compute $f$ with negligible error (nominally set to $1/3$ and can be reduced to any other fixed value without changing the order of magnitude) by exchanging (noiseless) bits using some interactive protocol. The length of the shortest protocol achieving this is called the \textit{communication complexity} of $f$, and denoted by $CC(f)$. In the interactive communication setup, Alice and Bob must achieve their goal by communicating through a pair of independent BSC($\BSCprob$). The minimal length of an interactive {protocol} attaining this goal is now denoted by $CC_{\BSCprob}(f)$. 

In \cite{kol2013interactive}, Kol and Raz defined the interactive capacity as 
\begin{align}\label{eq:kolrazCi}
\Cinter^{\mathsf{KR}}(\BSCprob) \triangleq \lim_{\protLength\to\infty}\min_{f:CC(f)=n}\frac{\protLength}{CC_{\BSCprob}(f)},
\end{align}
and proved that 
\begin{align}\label{eq:kolrazrate}
\Cinter^{\mathsf{KR}}(\BSCprob) \geq  1-O(\sqrt{h(\BSCprob)})
\end{align}
 in the limit of $\BSCprob\to 0$, under the additional assumption that the communication complexity of $f$ is computed with the restriction that the order of speakers is predetermined and has some fixed period. The former assumption on the order of speakers is important. Indeed, consider again the example where the function $f$ is either Alice's input or Bob's input as decided by Alice. In this case, the communication complexity with a predetermined order of speakers is double that without this restriction, and hence considering such protocols renders $\Cinter^{\mathsf{KR}}(\BSCprob) \leq \frac{1}{2}$. For further discussion on the impact of speaking order as well as channel models that allow collisions, see \cite{haeupler2014interactive}. For a fixed nonzero $\BSCprob$, the coding scheme presented in \cite{schulman1992communication} (which precedes \cite{kol2013interactive}) already showed that $\Cinter^{\mathsf{KR}}(\BSCprob)  = \Theta(\Cshannon(\BSCprob))$, but the constant has not been computed. In \cite{InteractiveLowerBound} it is shown that 	${\Cinter(P_{Y|X})}\geq  0.0302\cdot {\Cshannon(P_{Y|X})}$ for any channel $P_{Y|X}$ taken from the class of binary memoryless symmetric channels (which include the binary symmetric channel, the binary erasure channel, the binary input additive Gaussian channel etc.).
 

\section{Finite-state Protocols\label{section:finitestateformulation}}
Let us start by defining the notions of interactive rate and capacity for families of protocols. Let $\bs{\Pi}=\{\Pi_1,\Pi_2,...\}$ be a sequence of families of protocols, 
where $\Pi_n$ denotes some family of length-$n$ protocols. A rate $R$ is called \textit{achievable} for $\bs{\Pi}$ if there exists a sequence $\simulatingprotocol_{\protLength}$ of $(\protLength,\sigmalength_{\protLength})$ coding schemes where $\sigmalength_{\protLength} \leq \frac{\protLength}{R}$, 
and such that 
\begin{align}\label{eq:rateerrorPin}
\lim_{n\to \infty}\max_{\pi\in \Pi_n} P_e(\Sigma_n,\bs{\pi})  = 0.
\end{align}
Namely, the difference from 
 \eqref{eq:rateerror} is that now the maximum is taken over the protocols in $\Pi_n$ and not over the entire family of protocols with length $\protLength$.
Accordingly, we denote the interactive capacity respective to the channel $P_{X|Y}$ and the family of protocols $\bs{\Pi}$ by $\Cinter(\bs{\Pi},P_{Y|X})$, and define it as the maximum of all achievable rates for $P_{Y|X}$ and $\bs{\Pi}$.
 
The family of protocols studied in this paper is the family of finite-state protocols with $\StateOrder$ states, which will be referred to in short as \textit{$\StateOrder$-states} protocols. In these protocols, the entire history of the transcript is encapsulated in a state-variable taken from a set with a finite cardinality. The state variable determines the following transcript bit, and is advanced by both parties using a predetermined update rule.

The notation of finite-state protocols is given here:
\begin{definition}\label{def:finitestate}
Let $\bs{\FiniteStateProtocol}_{\StateOrder}=\{\FiniteStateProtocol_{\StateOrder,1},\FiniteStateProtocol_{\StateOrder,2},...\}$ denote the family of $\StateOrder$-state protocols of increasing lengths. For these protocols Alice speaks at odd times and Bob speaks on even times: namely $\sigma_\tind=\mathrm{Alice}$ if $\tind$ is odd, and $\sigma_\tind=\mathrm{Bob}$ if $\tind$ is even. The transcript of these protocols is generated by
\begin{align}\label{eq:finitenextbit}
\protTrans_{\tind} = {\finitestatefunc}_{\tind}(\statevar_{\tind-1}),
\end{align}
where $\statevar_{\tind}$ is the state variable at time $\tind$, $\statevar_{\tind}\in \statespace$ . $\statespace$ is the state-space, with cardinality $|\statespace|=\StateOrder$ assumed to be $\statespace=\{0,1,...,\StateOrder-1\}$ without loss of generality. 
$\transFunc_{\tind}:\statespace\mapsto\{0,1\}$ is the transmission function at time $\tind$, owned by Alice at odd $\tind$ and by Bob at even $\tind$ and assumed to be unknown to the counterpart.
In addition, the state $\statevar_{\tind}$ is advanced in time according to 
\begin{align}\label{eq:finitenextstate}
\statevar_{\tind} = \stateadvancefunc(\statevar_{\tind-1},\protTrans_{\tind}),
\end{align}
where $\stateadvancefunc:(\statespace,\{0,1\})\mapsto\statespace$ is the state-advance function, which is time invariant and known to both parties.
\end{definition}
The following example for a finite-state protocols is the family of Markovian protocols previously presented in \cite{MarkovianISIT} and defined as follows:
\begin{example}\label{example:markovian}
For a Markovian protocol, the number of states $\StateOrder$ is a power of two, and the state variable corresponds to the last $\log\StateOrder$ bits of the transcript. Namely, the state can be regarded as the binary vector
\begin{align}
\bs{\statevar}_{\tind-1} = (\statevar_{\tind-1}(1),...,\statevar_{\tind-1}(\log {\StateOrder})) = ({\protTrans}_{\tind-\log {\StateOrder}},..., {\protTrans}_{\tind-1}).
\end{align}
and the state-advance function is
\begin{align}
\bs{\statevar}_{\tind}=\stateadvancefunc(\bs{\statevar}_{\tind-1},\protTrans_{\tind})=
(\statevar_{\tind-1}(2),...,\statevar_{\tind-1}(\log {\StateOrder}-1),\protTrans_{\tind}).
\end{align}
\end{example}

\section{Basic Concepts of the Coding Schemes \label{section:codingscheme}}
The proofs in this paper are based on constructive coding schemes which use the concept of vertical simulation presented below, implemented in conjunction with either one of the two methods described in Subsections~\ref{subsec:effstate} and \ref{subsec:effexhaust}.

\subsection{Vertical Simulation}
\begin{table*}
	\centering
\begin{tabular}{|c|l||l|l|l|l|l|}
\hline block \#& initial state &
\multicolumn{5}{ |c| }{transcript} \\
\hline\hline
$1$&$\statevar_0$ & $\protTrans_{1}$ & $\protTrans_{2}$&$\ldots$&$\protTrans_{\blockLength-1}$ &$\protTrans_{\blockLength}$\\\hline
$2$&$\statevar_{\blockLength}$ & $\protTrans_{\blockLength+1}$ & $\protTrans_{\blockLength+2}$&$\ldots$&$\protTrans_{2\blockLength-1}$ & $\protTrans_{2\blockLength}$\\\hline
$3$&$\statevar_{2\blockLength}$ & $\protTrans_{2\blockLength+1}$ & $\protTrans_{2\blockLength+2}$&$\ldots$&$\protTrans_{3\blockLength-1}$ & $\protTrans_{3\blockLength}$\\\hline
$\vdots$&$\vdots$&$\vdots$ & & & & $\vdots$\\\hline
$\protLength/\blockLength$  &
$\statevar_{\protLength-\blockLength}$ & $\protTrans_{\protLength-\blockLength+1}$ & $\protTrans_{\protLength-\blockLength+2}$&$\ldots$&$\protTrans_{\protLength-1}$ & $\protTrans_{\protLength}$\\ \hline
\multicolumn{2}{|c|}{speaker} & Alice & Bob &$\ldots$&Alice&Bob\\
\hline\hline
\multicolumn{2}{|c||}{vertical block \#}&$1$&$2$&$\cdots$&$\blockLength-1$&$\blockLength$\\ \hline
\end{tabular}
\vspace{2mm}
\caption{Vertical protocol simulation \label{table:blocktrans}}
\end{table*}

As explained before, the transcript bits of interactive protocols are produced sequentially ($\protTrans_1, \protTrans_2, \protTrans_3, ....$). Simulating a protocol over noisy channel requires the reliable transmission of the bits sent in every round, whose number is potentially small (and can even be equal to one, in the extreme case and in the finite-state protocols discussed in this paper), which impedes the use of efficient channel codes due to finite block-length bounds \cite{polyanskiy2010channel}.

One way of circumventing the problem of a short block-length (i.e. small number of bits per round) is using vertical simulation as explained in this subsection. The concept of a vertical simulation is depicted in Table~\ref{table:blocktrans}, in which the protocol is simulated in \textit{vertical blocks}, according to the indexing at the bottom row of the table. Namely, the first vertical block contains the transcript bits $(\protTrans_1,\protTrans_{\blockLength+1}, \protTrans_{2\blockLength+1},...)$, the second vertical block contains the 
transcript bits $(\protTrans_2,\protTrans_{\blockLength+2}, \protTrans_{2\blockLength+2},...)$ and so on. As shall be explained in the sequel, the vertical blocks can be constructed to be sufficiently long in order to allow reliable transmission at rates approaching Shannon capacity. The main obstacle of using this technique in the general case is the assumption that future transcript bits 
(for example $\protTrans_{\blockLength+1}$, $\protTrans_{2\blockLength+1}$ etc.~for the first vertical block)
are known prior to the simulation of the protocol.  In the sequel, we shall provide methods for the efficient computation of these future transcript bit, which facilitate the simulation of the entire protocol at Shannon capacity.

Let us now explicitly define the concept of vertical simulation. Let the $\protLength$ times of the protocols be divided into $\protLength/\blockLength$ blocks of length $\blockLength$, and assume that all the \textit{initial-states} respective to the beginnings of all blocks ($\statevar_0, \statevar_{\blockLength}, \statevar_{2\blockLength}$ etc.), are known to both parties before the transcript is simulated. For simplicity of presentation, one can consider at this point that the initial states are calculated and revealed by a genie, who knows the transmission functions of both parties. More realistic methods for calculating the initial states are elaborated later in this section. We now note, that by the finite-state property in Definition~\ref{def:finitestate}, having the initial-states of all the blocks known, the parties can continue simulating the transcript of every block, without needing to know the transcripts of its preceding blocks. In other words, the knowledge of the initial state at every blocks decouples the simulation problems of distinct blocks.

Using this decoupling assumption, the following coding scheme can be used for the simulation of the protocol over $P_{Y|X}$. We start by defining the vectors of state estimates and transcript estimates held by Alice and Bob. We use a distinct notation for every party and emphasize the fact that these are estimates, since they are computed over noisy channels. We denote the vector of initial state estimates at Alice's side for vertical block $j$ by
\begin{align}\label{eq:statevecdef}
\hat{\bs{\statevar}}^{A}(j)\dfn(\hat{\statevar}^A_{j-1}, \hat{\statevar}^A_{j+\blockLength-1},\hat{\statevar}^A_{j+2\blockLength-1},...,\hat{\statevar}^A_{j+\protLength-\blockLength-1}),    
\end{align}
and the respective vector of transcript estimates by
\begin{align}\label{eq:transvecdef}
\hat{\bs{\protTrans}}^A(j)\dfn(\hat{\protTrans}^A_j,\hat{\protTrans}^A_{j+\blockLength},
\hat{\protTrans}^A_{j+2\blockLength}
...,\hat{\protTrans}^A_{j+\protLength-\blockLength}).
\end{align}
Bob's counterparts to $\hat{\bs{\statevar}}^{A}(j)$ and $\hat{\bs{\protTrans}}^A(j)$ are respectively denoted by $\hat{\bs{\statevar}}^{B}(j)$ and $\hat{\bs{\protTrans}}^B(j)$ and are similarly defined. The scheme can now be presented for odd $j$ from $1$ to $\blockLength$:
\begin{enumerate}
    \item Assume that Alice and Bob have $\hat{\bs{\statevar}}^{A}(j)$ and $\hat{\bs{\statevar}}^{B}(j)$.
    \item Alice uses $\hat{\bs{\statevar}}^{A}(j)$ to calculate  $\hat{\bs{\protTrans}}^A(j)$ according to \eqref{eq:finitenextbit}.
    \item Alice encodes $\hat{\bs{\protTrans}}^A(j)$ using a block code with rate $R_v<\Cshannon(P_{Y|X})$, and sends it to Bob over the channel, using  $\frac{\protLength/\blockLength}{R}$ times. This code will be referred to as a \textit{vertical} block code.
    \item Bob decodes the output of the channel and obtains $\hat{\bs{\protTrans}}^B(j)$.
    \item Alice (resp. Bob) uses $\hat{\bs{\statevar}}^{A}(j)$ and $\hat{\bs{\protTrans}}^A(j)$ (resp. $\hat{\bs{\statevar}}^{B}(j)$ and $\hat{\bs{\protTrans}}^B(j)$) to calculate $\hat{\bs{\statevar}}^{A}(j+1)$ (resp. $\hat{\bs{\statevar}}^{B}(j+1)$) according to \eqref{eq:finitenextstate}.
    \item Alice and Bob advance $j$ by one.
\end{enumerate}
For even $j$, the same steps are implemented by exchanging the roles of Alice and Bob.
We recall that we previously assumed that for the first block, both parties know the actual initial states of the noiseless protocol, i.e. $\hat{\bs{\statevar}}^{A}(1)=\hat{\bs{\statevar}}^{B}(1)$ and both are equal to the state vector of the noiseless protocol. It is clear from the construction of the scheme, that if all block codes are reliably decoded, the transcript is simulated without error. The following basic lemma gives a condition for the reliable decoding of block codes:
 \begin{lemma}\label{lemma:blockwise} Suppose $l(n)$ independent blocks of $b(n)$ bits are to be conveyed over channel $P_{Y|X}$ at rate $R<\Cshannon(P_{Y|X})$ and $n\to \infty$. Then, if $l(n) = o(e^{b(n)})$, the probability of error in the decoding of one or more blocks is $o(1)$. 
 \end{lemma}
 The proof is due to the basic fact that the probability of error decays exponentially in the block length and appears in 
 Appendix~\ref{appendix:lemma}. 
 
 From this point on, we set $\blockLength=\sqrt{\protLength}$. We assume that if needed, the transcript is extended by zeros in order to ensure that $\sqrt{\protLength}$ is an integer. Using  Lemma~\ref{lemma:blockwise} 
 with $l(n)=m(n)=\sqrt{m}$ 
 ensures that reliable transmission of the vertical blocks can be accomplished at any rate $R_v<\Cshannon(P_{Y|X})$.
 
 Let us now bound the total length $\sigmalength$ of the simulating protocol: 
 \begin{align}
     \sigmalength = \frac{\protLength}{\blockLength}\frac{\blockLength}{R_v}=\frac{\protLength}{R_v}.
 \end{align}
 Therefore $\frac{\protLength}{\sigmalength}= R_v$ for every $R_v<\Cshannon(P_{Y|X})$, which means that the protocol can be reliably simulated at Shannon capacity if $\protLength\to\infty$.
 
 So far, we assumed without justification, that initial states of all the blocks were revealed to both parties before the beginning of the simulation. We now present two alternative methods for their efficient calculation.
 

 \subsection{Efficient State  Lookahead\label{subsec:effstate}}
 This method is based on two assumptions:
 \begin{enumerate}
     \item For every block, the last state can be calculated by both parties given the first state,  without knowing the entire transcript of the block, 
     using only $o(\blockLength)$ (clean) bits exchanged between the parties.
     \item The $\frac{\protLength}{\blockLength}o(\blockLength)$ bits required for this calculation for the entire protocol, can be reliably exchanged over the noisy channels at a strictly positive rate.
 \end{enumerate}
Assuming that the very first state of the protocol is known to both parties, and that the first condition holds, Alice and Bob can go from the first block to the last and calculate all their respective initial states.
The second condition guarantees that only additional $\Theta(\frac{\protLength}{\blockLength}o(\blockLength))=o(\protLength)$ channel uses are required for this process. The total length of the simulating protocol can thus be bounded by 
 \begin{align}
     \sigmalength \leq \frac{\protLength}{\blockLength}\frac{\blockLength}{R_v}+o(\protLength),
 \end{align}
 so, as before, $\lim_{\protLength\to\infty}\frac{\protLength}{\sigmalength}= R_v$ for every $R_v<\Cshannon(P_{Y|X})$, which means that the protocol can be reliably simulated at Shannon capacity provided that $\protLength\to\infty$.
 


\subsection{Efficient Exhaustive Simulation \label{subsec:effexhaust}}
The following method was previously presented in \cite{MarkovianISIT} for the simulation of Markovian protocols. 
So far, we assumed that for every block, only the transcript related to single initial state, which was assumed to be the actual state in the noiseless protocol, was simulated. Alternatively, it is possible to simulate all transcripts resulting from all possible initial states in every block, and then go from the first block to the last and estimate the transcript of the noiseless protocol according the the final state of the previous block. 
Such a simulation can be made possible, for example, if the parties simply describe the identities of their transmission functions to their counterparts. While it is easy to show that the required bits can be conveyed at Shannon capacity, if there are more than two possible transmission functions at every time, the total rate of such a coding scheme is bound to be lower than Shannon capacity.

However, Shannon capacity can be achieved if the following conditions hold:
\begin{enumerate}
    \item At every block, the transcripts associated with all possible $\StateOrder$ initial states, can be encoded using only $\blockLength+o(\blockLength)$ bits.
    \item The required bits can be reliably conveyed over the noisy channels at any rate below Shannon capacity.
\end{enumerate}
If both conditions hold then the total number of channel usese required for the simulation is
 \begin{align}
     \sigmalength \leq \frac{\protLength}{\blockLength}\frac{\blockLength+o(\blockLength)}{R_v},
 \end{align}
 and the protocol can be simulated at any rate below Shannon capacity as long as $\protLength\to \infty$.


\section{Achieving Shannon Capacity with Two States \label{section:twostates}}
The first result presented in this paper is that any two-state protocol can be simulated at Shannon capacity. An equivalent statement is given in the following theorem:
\begin{theorem} 
\begin{align}
    \Cinter(\bs{\Pi},P_{Y|X})=\Cshannon(P_{Y|X}),
\end{align}
where $\bs{\Pi}=\bs{\FiniteStateProtocol}_2$, namely, the family of two-state protocols.
\end{theorem}
The proof is based on the following coding scheme:
\begin{proof}
We assume without loss of generality that $\statespace=\{0,1\}$ and start by presenting an algorithm for the {efficient state lookahead} method from Subsection~\ref{subsec:effstate}. For simplicity of exposition we use the time indices of the first block. For other blocks the indices should be appropriately shifted. 
We also assume that the bits required for the algorithm are exchanged between Alice and Bob without error. In the sequel we explain how they can be reliably conveyed over the noisy channels. 

The first step in the algorithm is the calculation of the following sequence of \textit{composite-functions},
$\stateiterfunc_\tind:\statespace\mapsto\statespace$, defined as:
    \begin{align}\label{eq:nufunc}
        \stateiterfunc_\tind(\statevar_{\tind-1})\dfn \stateadvancefunc(\statevar_{\tind-1},\finitestatefunc_{\tind}(\statevar_{\tind-1})),
    \end{align}
for $1\leq i\leq\blockLength$, which is done by Alice at odd $\tind$ and Bob at even $\tind$.  We note that knowing $\stateiterfunc_\tind(\statevar_{\tind-1})$, and the value of $\statevar_{\tind-1}$, the following state $\statevar_{\tind}$ can be calculated. We also note that since $\stateiterfunc_\tind:\{0,1\}\mapsto\{0,1\}$,  $\stateiterfunc_\tind(\statevar_{\tind-1})$ must be one of the following four functions:
    \begin{align}
    \stateiterfunc_\tind(\statevar_{\tind-1})=\statevar_{\tind-1}\oplus 0 , \quad
    \stateiterfunc_\tind(\statevar_{\tind-1})=\statevar_{\tind-1}\oplus 1, \quad \stateiterfunc_\tind(\statevar_{\tind-1})=0, \quad
    \stateiterfunc_\tind(\statevar_{\tind-1})=1,
    \end{align}    
    which can also be described in the following form:
    \begin{align}\label{eq:transfuncs_b_c}
    \stateiterfunc_\tind(\statevar_{\tind-1})=\statevar_{\tind-1}\oplus c_\tind  \text{ or } \stateiterfunc_\tind(\statevar_{\tind-1})=b_\tind,
    \end{align}
    where $b_\tind,c_\tind\in\{0,1\}$. 
    The basic idea of the algorithm is the following. If for all $1\leq\tind\leq \blockLength$ the composite-functions are $\stateiterfunc_\tind(\statevar_{\tind-1})=\statevar_{\tind-1}\oplus c_\tind$, then the final state $\statevar_{\blockLength}$ can be calculated by:
    \begin{align}
    \statevar_{\blockLength}&=
    \statevar_0\oplus\left[
    \sumparity_{\tind={1}}^{\blockLength} c_\tind\right]\\
    &= \statevar_0\oplus d_{\text{Alice}}\oplus d_{\text{Bob}}
    \label{eq:stateiter1}
    \end{align}
where
\begin{align}
    d_{\text{Alice}}&\dfn \sumparity_{\tind \text{ is odd},\tind\in \{{1},...,\blockLength\}} c_\tind\\
    d_{\text{Bob}}&\dfn \sumparity_{\tind \text{ is even},\tind\in \{{1},...,\blockLength\}} c_\tind.
    \end{align}    
    In other words, $\statevar_\blockLength$ can be calculated by its initial value $\statevar_0$ and the parity of the number of times in the block it is flipped (from $0$ to $1$ or vice versa) by either Alice or Bob. All in all, assuming that the parties know $\statevar_0$, they only need to exchange $d_{\text{Alice}}$ and $d_{\text{Bob}}$ (i.e. two bits) in order to calculate $\statevar_\blockLength$. 
    However, so far we assumed that all the compsite functions in the block in the following form $\stateiterfunc_\tind(\statevar_{\tind-1})=\statevar_{\tind-1}\oplus c_\tind$. In the general case in which $\stateiterfunc_\tind(\cdot)$ are taken from the complete set of four functions in \eqref{eq:transfuncs_b_c}, the algorithm can be modified by first exchanging the location and the value of the last \textit{constant} composite-function in the block, i.e. the last composite-function of the form $\stateiterfunc_\tind(\statevar_{\tind-1})=b_\tind$. We note that this process requires only exchanging $O(\log\blockLength)$ between Alice and Bob. Then, $\statevar_{\blockLength}$ can be calculated similarly to  \eqref{eq:stateiter1} but from the location of the last constant composite-function and not from the beginning of the block.
    
    The algorithm is formulated as follows:
\begin{enumerate}
    \item\label{step2} Alice sends Bob her latest (odd) time index in the block for which $\stateiterfunc_\tind(\statevar_{\tind-1})=b_{\tind}$, $b_{\tind}\in\{0,1\}$ (i.e. her latest constant composite-function), along with value of $b_{\tind}$. If such an index does not exist she sends zero to Bob. Bob then repeats the same process with the appropriate alterations. 
    We use $\tind_{\text{const}}$ to denote the maximum of the indices, which therefore represents the location of the last constant composite-function in the block. We now set $b_{0}=\statevar_{0}$ if $\tind_{\text{const}}=0$ and 
    $b_{\tind_{\text{const}}}$ if $\tind_{\text{const}}>0$. 
    This process requires exchanging $O(\log\blockLength)$ bits between Alice and Bob.
    \item\label{step3} We now note, that since $\tind_{\text{const}}$ is the index of the latest constant composite-function in the block, then for all $\tind_{\text{const}}<\tind\leq\blockLength$,  $\stateiterfunc_\tind(\statevar_{\tind-1})=\statevar_{\tind-1}\oplus c_{\tind}$ for some $c_{\tind}\in\{0,1\}$. The final state in the block, $\statevar_{\blockLength}$, can therefore be calculated by
    \begin{align}
    \statevar_{\blockLength}&=b_{\tind_{\text{const}}}\oplus
    \sumparity_{\tind=\tind_{\text{const}+1}}^{\blockLength} c_\tind\\
    &=b_{\tind_{\text{const}}}\oplus d_{\text{Alice}}\oplus d_{\text{Bob}}
    \label{eq:statelookahead}
    \end{align}
where
\begin{align}
    d_{\text{Alice}}&\dfn \sumparity_{\tind \text{ is odd},\tind\in \{\tind_{\text{const}+1},...,\blockLength\}} c_\tind\\
    d_{\text{Bob}}&\dfn \sumparity_{\tind \text{ is even},\tind\in \{\tind_{\text{const}+1},...,\blockLength\}} c_\tind.\\    
\end{align}
We finally note, that $d_{\text{Alice}}$ and $d_{\text{Bob}}$ are single bits that can be calculated by their respective parties and then exchanged, leaving the total number of required exchanged bits for the algorithm $O(\log\blockLength)$.
\end{enumerate}
After repeating this operation for all blocks, it is possible 
to calculate all the final states of all blocks (i.e.~all the initial states of their following blocks) 
by applying \eqref{eq:statelookahead} from the first block to the last.

It only remains to verify that the respective $O(\log\blockLength)$  bits per block can be reliably conveyed over the noisy channels between Alice and Bob using the channel times $o(\protLength)$ times as required in Subsection~\ref{subsec:effstate}. This task can be easily performed, for example by using one block code per party containing $O(\frac{\protLength}{\blockLength}\log\blockLength)=O(\sqrt{\protLength}\log\protLength)$ bits.
\end{proof}
For the sake of completeness we now give the high level of an alternative coding scheme based on the efficient exhaustive simulation method described in Subsection~\ref{subsec:effexhaust}. This coding scheme is a little more involved than the previously described one, and depends on the identity of the state-advance function $\stateadvancefunc(\cdot)$.
We start by noting that $\stateadvancefunc(\cdot)$ is a binary function with two binary inputs, so there are in total sixteen possible such function. In particular, there are four state-advance function that do not depend on transcript bit $\protTrans_{\tind}$:
\begin{align}
    \stateadvancefunc(\statevar_{\tind-1},\protTrans_{\tind})=0\oplus \statevar_{\tind-1},\quad 
    \stateadvancefunc(\statevar_{\tind-1},\protTrans_{\tind})=1\oplus \statevar_{\tind-1},\quad
    \stateadvancefunc(\statevar_{\tind-1},\protTrans_{\tind})=0,\quad 
    \stateadvancefunc(\statevar_{\tind-1},\protTrans_{\tind})=1.
\end{align}
As the very first state of the protocol is assumed to be known to both parties, having one of these state-advance functions, the state sequence of the entire protocol can be determined before its simulation, rendering the entire protocol non-interactive, hence trivial to simulate. For the remaining twelve state-advance functions, the following coding scheme is proposed, which is described for simplicity for the first block, but should be independently implemented for all blocks:
\begin{enumerate}
    \item Before the simulation begins, both parties communicate the locations of the first (rather than the last) constant composite-function in the block: the smallest value  $1\leq\tind\leq\blockind$ for which $\stateiterfunc_\tind(\statevar_{\tind-1})=b_{\tind}$, for some $b_{\tind}\in\{0,1\}$. This process requires exchanging $O(\log\blockLength)$ bits.
    \item The parties exchange the identities of their transmission functions (i.e. $\finitestatefunc\tind(\cdot)$) before the location of the first constant composite-function in the block, using a single bit per time index. In the sequel we show that there are indeed only two relevant functions to describe, so their description requires only a single bit.  At the end of this process, the parties can independently simulate the transcripts for both initial states until the location of the first constant composite-function.
    \item For time indices after the location of the first constant composite-function, the transcripts associated with both initial states coincide, so they can both be simulated using a single bit per time index.
\end{enumerate}
Using this coding scheme, only $\blockLength+o(\blockLength)$ bits are required for the simulation of the transcripts associated with both initial states. These bits can be reliably conveyed for all blocks using vertical block codes, as required in the description of the scheme in Subsection~\ref{subsec:effexhaust}.

To see that, observe that there are only three canonical types of state-advance functions, depicted in the state-diagrams in Figure~\ref{fig:stateadvance}. The nodes represent the state variables, and the directed edges show the possible state transitions. The specific values of the states and transcript bits on the edges are deliberately not indicated; it is easy to check that there are four possible setting for every type, summing up to twelve functions in total. An example for a Type I state-advance function is $\stateadvancefunc(\statevar_{\tind-1},\protTrans_{\tind})=\protTrans_{\tind}$,
for a Type II state-advance function is $\stateadvancefunc(\statevar_{\tind-1},\protTrans_{\tind})=\statevar_{\tind-1}\wedge\protTrans_{\tind}$, and for a Type III state-advance function is:
$\stateadvancefunc(\statevar_{\tind-1},\protTrans_{\tind})=1\oplus(\statevar_{\tind-1}\wedge\protTrans_{\tind})$. We now return to the definition of the composite-functions $\stateiterfunc_\tind(\statevar_{\tind-1})$ in \eqref{eq:nufunc}, and note that $\stateiterfunc_\tind(\statevar_{\tind-1})$ is constant (i.e. set to either $0$ or $1$) if the transmission functions are such that $\statevar_{\tind}$ receives the same value for both $\statevar_{\tind-1}=0$ and $\statevar_{\tind-1}=1$. As the transmission function $\finitestatefunc_{\tind}(\statevar_{\tind-1})$ determines the values associated with the edges of the state diagram, it can be seen that for every type of advance function, there exist only two transmission functions which render $\stateiterfunc_\tind(\statevar_{\tind})$ constant. Since there are in total four possible transmission function, there are therefore only two possible transmission functions before the appearance of one of the two that makes 
$\stateiterfunc_\tind(\statevar_{\tind})$ constant, as required by the scheme.

\begin{figure}[ht]
\centering
\includegraphics{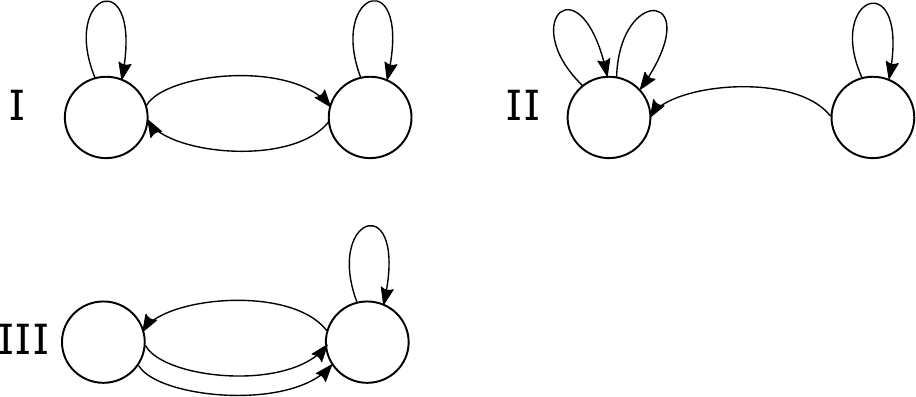}
\caption{State diagrams of the three types of state-advance functions.
\label{fig:stateadvance}}
\end{figure}


\section{Failure of the Coding Scheme for Three States\label{section:threestates}}
\begin{figure}
	\centering
	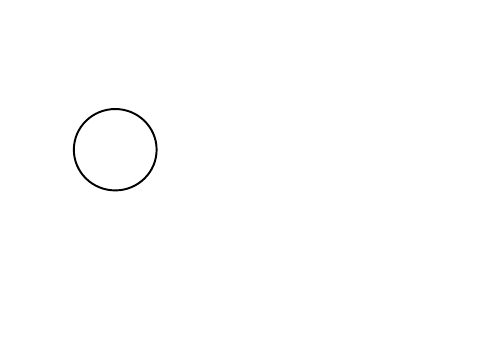
	\caption{A diagram representing the state-advance function $\stateadvancefunc(\statevar_{\tind-1},\protTrans_{\tind})$ of the three-state protocol from Section~\ref{section:threestates}. The nodes represent the states, and an edge $(\statevar_{\tind-1}, \statevar_{\tind})$ represents a transition from state $\statevar_{\tind-1}$ to $\statevar_{\tind}$. The numbers attached to the edges are the transcript bit $\protTrans_\tind={\finitestatefunc}_{\tind}(\statevar_{\tind-1})$. \label{fig:threestate}} 
\end{figure}
We now provide an example of a protocol for which both methods described in Subsections~\ref{subsec:effstate} and \ref{subsec:effexhaust} fail. Since the protocol is to be used on a block (rather than on the entire protocol), we use $\blockLength$ to denote its length.
\begin{example}\label{example:disjprotocol}
We define the following interactive three-state protocol ($\statespace=\{0,1,2\}$) of length $\blockLength$. The state-advance rule $\statevar_{\tind} = \stateadvancefunc(\statevar_{\tind-1},\protTrans_{\tind})$ is depicted in Figure~\ref{fig:threestate}. Namely, at state $\statevar_{\tind-1}=0$, the next state is $\statevar_{\tind}=\protTrans_{\tind}$. At state $\statevar_{\tind-1}=1$ the next state is $\statevar_{\tind}=0$ if $\protTrans_{\tind}=0$ and $\statevar_{\tind}=2$ if $\protTrans_{\tind}=1$. At state $\statevar_{\tind-1}=2$, the next state is $\statevar_{\tind}=2$ regardless the value of $\protTrans_{\tind}$.

Let $\bs{\alpha}=(\alpha_1,\alpha_2,...\alpha_\blockLength)$ and $\bs{\beta}=(\beta_1,\beta_2,...\beta_\blockLength)$ be binary sequences. Assume Alice knows the elements of both sequences only at odd indices, and Bob knows the elements of both sequences only at even indices. The following transmission function is used by Alice at odd time indices:
\begin{align}
\protTrans_{\tind} = {\finitestatefunc}_{\tind}(\statevar_{\tind-1})
=\begin{cases}
\alpha_\tind & \text{if } \statevar_{\tind-1}\in\{0,1\}\\
\beta_\tind & \text{if } \statevar_{\tind-1}=2,\\
\end{cases}
\end{align}
and the following transmission function is used by Bob at even time indices:
\begin{align}
\protTrans_{\tind} = {\finitestatefunc}_{\tind}(\statevar_{\tind-1})
=\begin{cases}
\alpha_\tind\wedge\statevar_{\tind-1}  & \text{if } \statevar_{\tind-1}\in\{0,1\}\\
\beta_\tind & \text{if } \statevar_{\tind-1}=2.\\
\end{cases}
\end{align}
\end{example}
We start by proving the failure of the efficient state lookahead scheme from Subsection \ref{subsec:effstate}, by showing a reduction from the disjointness problem commonly used in the communication complexity literature \cite{communicationComplexityBook}.  
\begin{definition}[Disjointness]
Alice and Bob are given as input the sets $X,Y\subseteq \{1,...,\blockLength/2\}$, respectively. The disjointness function is defined as
\begin{align}
    \mathrm{DISJ}(X,Y)\dfn \indfunc{X\cap Y=\emptyset},
\end{align}
where $\indfunc{\cdot}$ is the indicator function, which equals one if the condition is satisfied and zero otherwise.
\end{definition}
We now show how $\mathrm{DISJ}(X,Y)$ can be computed using the three-state protocol of Example~\ref{example:disjprotocol}. We set the values of the vector $\bs{\alpha}$ for $k\in\{1,2,...,\blockLength/2\}$ according to
\begin{align}
    \alpha_{2k-1} = \indfunc{k\in X},\quad \alpha_{2k} = \indfunc{k\in Y}.
\end{align}
The values of the elements of $\bs{\beta}$ do not affect the reduction from disjointness and can be all set to zero for simplicity. They will be used in the proof of the failure of the exhaustive simulation scheme shown in the sequel.

Observe that $s_\blockLength=2$ if and only if there exist at least one  $k\in\{1,2,...,\blockLength/2\}$ for which
$\alpha_{2k-1}=1$ and $\alpha_{2k}=1$, which means that 
${k\in X}$ and ${k\in Y}$ and the intersection of $X$ and $Y$ is not empty. Namely,
\begin{align}\label{eq:disjointnessreduction}
    \mathrm{DISJ}(X,Y)=\indfunc{s_\blockLength\in\{0,1\}},
\end{align}
which means that $\statevar_\blockLength$ can be used to compute $\mathrm{DISJ}(X,Y)$.

Since it is assumed that $\statevar_\blockLength$ can be computed using $o(\blockLength)$ bits, and $\mathrm{DISJ}(X,Y)$ can be computed using $s_\blockLength$ without additional communication due to \eqref{eq:disjointnessreduction}, it follows that $\mathrm{DISJ}(X,Y)$ can also be computed using $o(\blockLength)$ bits. However, it is well-known that the communication complexity of the disjointness function is $\Omega(\blockLength)$:
even if Alice and Bob can use a shared randomness source in their communication protocol, and even if they are allowed to err with probability $1/3$, they must still exchange $\Omega(\blockLength)$ bits in the worst-case in order to compute $\mathrm{DISJ}(X,Y)$ \cite{kalyanasundaram1992probabilistic,razborov1990distributional}.
In fact, disjointness remains hard even in the \emph{amortized} case, where Alice and Bob are given a sequence of inputs $X_1,...,X_l \subseteq \{1,...,\blockLength/2\}$ and $Y_1,...,Y_l \subseteq \{1,...,\blockLength/2\}$ (respectively),
and their goal is to output the sequence $\mathrm{DISJ}(X_1,Y_1),...,\mathrm{DISJ}(X_l,Y_l)$. The \emph{average} communication per-copy for this task is still $\Omega(m)$ (i.e., the \emph{total} communication is $\Omega(m \cdot l$, where $l$ is the number of copies). This result is the direct consequence of the following three results: i) the information cost of disjointness is linear \cite{bar2004information}; ii) information cost is additive \cite{braverman2014information}; and iii) information is a lower bound on communication \cite{bar2004information}.

We now prove that for Example~\ref{example:disjprotocol}, the efficient exhaustive simulation of Subsection~\ref{subsec:effexhaust} also fails, by providing a setting of $\bs{\alpha}$ and $\bs{\beta}$ for which simulating the transcripts of all three possible initial states requires the parties to reliably exchange $\frac{3}{2}\blockLength$ bits. This is impossible to accomplish using only $\blockLength+o(\blockLength)$ exchanged bits, as assumed in the scheme, and therefore the scheme must fail.
We set up the example as follows: we set $\bs{\beta}$ to be an arbitrary binary vector whose odd elements are known only to Alice and whose even elements are known only to Bob. In addition, we set the odd elements of $\bs{\alpha}$ to be arbitrary and known only to Alice, and set all the even elements of $\bs{\alpha}$ to zero.
We observe that the transcript associated with the initial state $\statevar_0=2$ essentially sends the sequence $\beta_1,...,\beta_\blockLength$ non-interactively, implying that the parties exchange $\blockLength$ bits that are initially unknown to their counterparts. For the other two initial states, $\statevar_0 \in \{0,1\}$, 
the setting of the even entries of $\alpha$ to zero guarantees that the transcripts associated with the two initial states $\statevar_0 \in \{0,1\}$ will never reach the state $\statevar_i= 2$ at $i\leq\blockLength$.
This way, in order to simulate the associated transcripts, Alice must  convey to Bob the even entries of $\bs{\alpha}$. Hence, 
successful exhaustive simulation means that $\frac{3}{2}\blockLength$ bits must be exchanged overall, which cannot be done using only $m+o(m)$ bits.


\section{Achieving Shannon Capacity with More Than Two States
\label{section:manystates}}
In the previous sections we presented a coding scheme that achieves capacity for {all} two-state protocols, but fails to achieve capacity for at least one three-state protocol. In this section we present specific families of $\StateOrder$-state protocols which obey two conditions, and show that within these families, almost all protocols can be simulated at Shannon capacity. 
The notion of achieving capacity for almost all members of a family is demonstrated in the following example:
\begin{example}
Consider the family of Markovian protocols with $\StateOrder$ states defined in Example~\ref{example:markovian} where $\StateOrder$ is a power of two, and whose transmission functions are taken from the entire set of $\statespace\mapsto\{0,1\}$ functions. We shall now show that capacity is achievable for almost all protocols in this family.

To see this, first observe that the set of possible transmission functions contains two \textit{constant} functions: one that maps all states to $0$ and one function that maps all states to $1$. Now, assume that vertical simulation is implemented as described in Section~\ref{section:codingscheme}, but all transcripts for initial states in all blocks are simulated for the last $\protLength^{1/4}$ times in every block (which requires only $o(\protLength)$) channel uses. It is easy to show (and a stronger statement is proved in Theorem~\ref{theorem:manystatesshannon} below) that almost all protocols in the family have at least one sequence of $\log\StateOrder$ constant functions within the last $\protLength^{1/4}$ times in every transmission block. Having this sequence of constant functions will ensure that all transcripts in every block will have the same final state, which could be used for the efficient state lookahead method described in Subsection~\ref{subsec:effstate}. 
\end{example}

On the other hand, one might argue that the presence of a sequence of constant transmission functions reduces the interactiveness of the protocol. In other words, that highly interactive protocols are not likely to include a constant function.
However, it was previously shown in \cite{MarkovianISIT}, that the scheme described above can be used for protocols whose transmission functions are taken from a smaller families of non-constant functions, such as the family of balanced Boolean functions. We shall now extend the results from 
\cite{MarkovianISIT} to finite-state protocols. For this purpose, we define two conditions that the family of protocols should fulfill. The first condition is related to the state-advance function, and the second  condition is related to the transmission functions, as defined here:
\begin{definition}\label{def:coincidingadvance}
A state-advance function $\stateadvancefunc$ of an
 $\StateOrder$-state protocol $\bs{\FiniteStateProtocol}_{\StateOrder}$ is called ``coinciding" if there exist $\coincidingConstant\in \mathbb{N}$ such that for every pair of distinct states $j,j'\in \statespace$, $j\neq j'$
 there exists a pair of binary sequences of length $\coincidingConstant$: $(b^j_1,b^j_2,...,b^j_{\coincidingConstant})$ and $(b^{j'}_1,b^{j'}_2,...,b^{j'}_{\coincidingConstant})$  for which 
  $\tilde{\statevar}^j_{\coincidingConstant}=\tilde{\statevar}^{j'}_{\coincidingConstant}$ where $\tilde{\statevar}^j_{\coincidingConstant}$ is generated by applying
 \begin{align}
 \tilde{\statevar}^j_{\tind} = \stateadvancefunc(\tilde{\statevar}^j_{\tind-1}, b^j_{\tind}),
 \end{align}
 for $\tind$ going from $1$ to $\coincidingConstant$ with the initial condition $\tilde{\statevar}^j_{0}=j$ and $\tilde{\statevar}^{j'}_{\coincidingConstant}$ is generated similarly, replacing $j$ by $j'$.
\end{definition}

\begin{definition}\label{def:usefulfuncs}
A set $\usefulset$ of $\StateOrder$-state transmission functions $\statespace\mapsto \{0,1\}$  is called ``useful" if for every pair of distinct states $\statevar,\statevar'\in\statespace$, $\statevar\neq\statevar'$  there exists at least one set of four functions  $\{f^{00},f^{01},f^{10},f^{11}\} \subseteq \usefulset$ for which
\begin{align}
f^{tt'}(\statevar)=t\text{ and }
f^{tt'}(\statevar')=t'
\end{align}
for all pairs $(t,t')\in\{0,1\}^2$.
\end{definition}
The following theorem formalizes the notion of achieving capacity for almost all members of these families of protocols:
\begin{theorem}\label{theorem:manystatesshannon}
Let $\bs{\Pi}$ be the family of all $M$-state protocols whose state-advance function is coinciding and whose transmission functions are taken from a fixed given set of useful functions. Then there exists a sequence of families of protocols $\bs{S}=\{S_1,S_2,...\}$,  $S_\protLength\subseteq\Pi_\protLength$ and $|S_\protLength|/|\Pi_\protLength|  = 1-o(1)$, for which the interactive capacity is equal to the Shannon capacity. Namely,
\begin{align}
    \Cinter(\bs{S},P_{Y|X})=\Cshannon(P_{Y|X}).
\end{align}
\end{theorem}

\begin{proof}
The proof is based on implementing the methods from Subsections~\ref{subsec:effstate} or Subsection~\ref{subsec:effexhaust} using one of the following two constructions. We start by presenting the construction for the efficient state lookahead method from 
Subsection~\ref{subsec:effstate}: 
For the last $\smallwindow=\protLength^{1/4}$ times in every transmission block,
exhaustively simulate the transcripts related to all possible $\StateOrder$ initial states. 
We assume for simplicity that the protocol is extended by zeros so that $\protLength^{1/4}$ is an integer and in addition, so that $\protLength^{1/4}/\coincidingConstant$ is also an integer, as shall be required in the sequel.
This simulation can be implemented by each side describing all its respective $\smallwindow/2$ transmission functions to its counterpart. After this is done, both parties can simulate the transcripts for all possible initial states in the last $\smallwindow$ times in every block without any additional channel uses. 
Since there are only $2^{\StateOrder}$ functions $\statespace\mapsto\{0,1\}$, the description of every function in $\usefulset$ requires at most ${\StateOrder}$ bits. The bits required for the description of all transmission functions of a party, for the last $\smallwindow$ times in all ${\protLength}/{\blockLength}$ transmission blocks, can be reliably conveyed over the noisy channel, either by a single block code per party, or by a distinct block code per time instance. It is easy to see that the setting of  $\smallwindow=\protLength^{1/4}$ and $\blockLength=\protLength^{1/2}$ ensures the transmission of these bits with a vanishing error using the channels only $o(\protLength)$ times. 
Now, if in every block, the transcripts respective to all possible initial state have the same (possibly block dependent) final state, we can use this set of states as the \textit{state lookahead}. We call this phenomenon \textit{state-coincidence} and note that if it occurs, since the channel was used only $o(\protLength)$ times for the calculation of the state lookahead, Shannon capacity can be achieved, as explained in Subsection~\ref{subsec:effstate}.

Alternatively, the efficient exhaustive simulation described in Subsection~\ref{subsec:effexhaust} can be similarly implemented by using the construction for the first (rather than the 
last) $\smallwindow$ times in all blocks. If all states coincide in all blocks, then for every block there is only a single transcript to simulate for the last $\blockLength-\smallwindow$ times in the block. All in all, only $\blockLength+o(\blockLength)$ bits are required for the simulation the transcripts of all the initial states, as required by the method.

We now use $S_\protLength\subseteq\Pi_\protLength$ to denote the subset of protocols for which the states coincide, so their respective transcripts can be simulated at Shannon capacity as explained above. It remains to prove that $|S_\protLength|/|\Pi_\protLength|  = 1-o(1)$. This is done by assuming that the protocols in $S_\protLength$ are generated by drawing all their transmission function uniformly from the set $\usefulset$ and independently in time, and denoting the probability of drawing a protocol in $S_\protLength$ by $\Pr(S_\protLength)$, so
\begin{align}
\Pr(S_\protLength)=\frac{|S_\protLength|}{|\usefulset|^{\protLength}}= 
\frac{|S_\protLength|}{|\Pi_\protLength|}.
\end{align}
Hence, proving that $\Pr(S_\protLength)=1-o(1)$ will prove the statement in the theorem.

We now show that indeed, the assumptions in the theorem ensure that $\Pr(S_\protLength)=1-o(1)$. We start by analyzing the probability of state-coincidence respective to a specific small block of length $\smallwindow$. For convenience, we assume that its time indices are $1$ to $\smallwindow$.
We start by denoting the transcript related to the initial state $j\in \statespace$ in by $(\protTrans^j_1,\protTrans^j_2,...,\protTrans^j_\smallwindow)$ and the respective sequence of states by $(\statevar^j_1,\statevar^j_2,...,\statevar^j_\smallwindow)$. More explicitly, 
$(\protTrans^j_1,\protTrans^j_2,...,\protTrans^j_\smallwindow)$ and $(\statevar^j_1,\statevar^j_2,...,\statevar^j_\smallwindow)$ are generated by the following iteration of \eqref{eq:finitenextbit}, \eqref{eq:finitenextstate}:
\begin{align}
\protTrans^j_{\tind} &= {\finitestatefunc}_{\tind}(\statevar^j_{\tind-1}),\label{eq:transupdate}\\
\statevar^j_{\tind} &= \stateadvancefunc(\statevar_{\tind-1},\protTrans^j_{\tind}),\label{eq:statesupdate}
\end{align}
for $\tind$ going from ${\tind}=1$ to ${\tind}=\smallwindow$ with the initial condition $\statevar^j_{0}=j$. We similarly define the transcript and the sequence of state respective to the initial state $j'\neq j$, $j'\in \statespace$ by $(\protTrans^{j'}_1,\protTrans^{j'}_2,...,\protTrans^{j'}_\smallwindow)$ and $(\statevar^{j'}_1,\statevar^{j'}_2,...,\statevar^{j'}_\smallwindow)$. We shall now bound the probability of state-coincidence: $\Pr(\statevar^{j}_\smallwindow=\statevar^{j'}_\smallwindow)$. 

Now, by the assumption that the state-advance function is coinciding (Definition~\ref{def:coincidingadvance}), there exist two binary sequences of $(b^j_1,b^j_2,...,b^j_{\coincidingConstant})$ and $(b^{j'}_1,b^{j'}_2,...,b^{j'}_{\coincidingConstant})$ for which 
$\tilde{\statevar}^j_{\coincidingConstant}=\tilde{\statevar}^{j'}_{\coincidingConstant}$ (the tilde in the notation of 
$\tilde{\statevar}^j_{\tind}$ and $\tilde{\statevar}^{j'}_{\tind}$ is used to distinguish them from ${\statevar}^j_{\tind}$ and ${\statevar}^{j'}_{\tind}$). The distinction is required, since $\tilde{\statevar}^j_{\tind}$ and $\tilde{\statevar}^{j'}_{\tind}$ are created by specific binary vectors, which in the general might not correspond to transcripts of protocols in $\Pi_n$.
We shall now use the coincidence and usefulness properties of the family of protocols in order to bound the probability of drawing a sequence of transmission functions, $\finitestatefunc_1,\finitestatefunc_2,...,\finitestatefunc_\coincidingConstant$, for which the iteration in \eqref{eq:transupdate}, \eqref{eq:statesupdate} yields
${\statevar}^j_{\coincidingConstant}={\statevar}^{j'}_{\coincidingConstant}$. We denote by $k$ the smallest time index for which
$\tilde{\statevar}^j_{k}=\tilde{\statevar}^{j'}_{k}$. It follows that for every $1\leq\tind\leq k$ we have $\tilde{\statevar}^j_{k}\neq\tilde{\statevar}^{j'}_{k}$. We now show that by the usefulness assumption, the binary sequences
$(b^j_1,b^j_2,...,b^j_{k})$, $(b^{j'}_1,b^{j'}_2,...,b^{j'}_{k})$ which generated the state sequences $(\tilde{\statevar}^j_{1},\tilde{\statevar}^j_{2},...,\tilde{\statevar}^j_{k})$, $(\tilde{\statevar}^{j'}_{1},\tilde{\statevar}^{j'}_{2},...,\tilde{\statevar}^{j'}_{k})$, can also be produced by a sequence of transmission functions $(\finitestatefunc,\finitestatefunc_2,...,\finitestatefunc_k)$ drawn uniformly and independently from $\usefulset$. The proof follows by observing that by the usefulness assumption, for every pair $(b^j_\tind,b^{j'}_\tind)$ for $1\leq\tind\leq k$, there exists at least one function $\finitestatefunc_\tind\in\usefulset$ such that
\begin{align}
    {\finitestatefunc}_{\tind}(\statevar)&=b^j_\tind,\label{eq:useful1}
    \text{ and}
    \\
    {\finitestatefunc}_{\tind}(\statevar')&=b^{j'}_\tind\label{eq:useful2}
\end{align}
for every $\statevar,\statevar'\in \statespace$, $\statevar\neq\statevar'$. In particular \eqref{eq:useful1} and \eqref{eq:useful2} also hold for the states in the sequences $(\tilde{\statevar}^j_{1},\tilde{\statevar}^j_{2},...,\tilde{\statevar}^j_{k})$, $(\tilde{\statevar}^{j'}_{1},\tilde{\statevar}^{j'}_{2},...,\tilde{\statevar}^{j'}_{k})$. Therefore, there exists a sequence of transmission functions $(\finitestatefunc,\finitestatefunc_2,...,\finitestatefunc_k)$, with $\statevar^j_k=\statevar^{j'}_k$ drawn with probability
\begin{align}
    \Pr\left(\statevar^j_{k}=\statevar^{j'}_{k}\right)\geq |\usefulset|^{-k}.
\end{align}
We now note that due to \eqref{eq:transupdate} and  \eqref{eq:statesupdate}, for every $k<\tind\leq\coincidingConstant$, we have that $\statevar^j_{\tind}=\statevar^{j'}_{\tind}$ for every choice of transmission functions $(\finitestatefunc_{k+1},...,\finitestatefunc_\coincidingConstant)$ and particularly for $\statevar^j_{\coincidingConstant}=\statevar^{j'}_{\coincidingConstant}$. Therefore,
\begin{align}
    \Pr\left(\statevar^j_{\coincidingConstant}=\statevar^{j'}_{\coincidingConstant}\right)\geq |\usefulset|^{-k}
    \geq |\usefulset|^{-\coincidingConstant}.\label{eq:success1}
\end{align}
We now observe, that  \eqref{eq:success1} only assumed that the initial states are distinct, i.e. $\statevar^j=j\neq\statevar^{j'}={j'}$. Therefore, in case $\statevar^j_{\coincidingConstant}\neq\statevar^{j'}_{\coincidingConstant}$
    we can consider the drawing of the following $\coincidingConstant$ functions as a repeated, statistically identical and independent experiment. Following this argumentation, we can consider consecutive $\smallwindow/\coincidingConstant$ such experiments, and observe that the a failure in the coincidence at the end of the block of length $\smallwindow$, $\statevar^j_{\smallwindow}\neq\statevar^{j'}_
    {\smallwindow}$, implies that all these $\smallwindow/\coincidingConstant$ experiments failed. We can therefore state the following bound:
    \begin{align}
    \Pr\left(\statevar^j_{\smallwindow}\neq\statevar^{j'}_
    {\smallwindow}
    \right)
    &\leq
    \left[\Pr\left(\statevar^j_{\coincidingConstant}\neq\statevar^{j'}_{\coincidingConstant}\right)\right]^{{\smallwindow}/{\coincidingConstant}}\label{eq:pfailure1}    
    \\
    &\leq (1-|\usefulset|^{-\coincidingConstant})^{{\smallwindow}/{\coincidingConstant}} \label{eq:pfailure2}\\
    &\leq \exp\left[
    -
    |\usefulset|^{-\coincidingConstant}
    {{\smallwindow}/{\coincidingConstant}}
    \right],\label{eq:pfailure3}
    \end{align}
    where \eqref{eq:pfailure1} is potentially loose since it considers only the coincidence events occurring in non-overlapping blocks of length $\coincidingConstant$,  \eqref{eq:pfailure2} is due to \eqref{eq:success1}, and finally  is by the inequality $(1-x)^a\leq \exp(-ax)$ which holds for any $x>0$ and $a\in\mathbb{N}$.

We emphasize that so far we examined the coincidence of only two initial states, $j$ and $j'$ in a single transmission block. We denote by $\mathcal{E}_1$ the event in which all the transcripts corresponding to all initial states did not coincide to the same final state. The probability of $\mathcal{E}_1$ can be bounded by:
\begin{align}
    \Pr(\mathcal{E}_1)
    &=
    \Pr\left(
    \bigcup_{j=1}^{\StateOrder-1}
    \statevar^0_{\smallwindow}\neq\statevar^{j}_
    {\smallwindow}
    \right)\\
    &\leq (\StateOrder-1)\exp\left[
    -
    |\usefulset|^{-\coincidingConstant}
    {{\smallwindow}/{\coincidingConstant}}
    \right]\label{eq:allstatescoincide}\\
    &<\StateOrder\exp\left[
    -
    |\usefulset|^{-\coincidingConstant}
    {{\smallwindow}/{\coincidingConstant}}
    \right]
\end{align}
where in \eqref{eq:allstatescoincide} we used the union bound and \eqref{eq:pfailure3}. Finally, we denote by $\mathcal{E}_2$ the event that the final states did not coincide in \textit{all} transmission blocks. Using the union bound again, this probability can be bounded by:
\begin{align}
        \Pr(\mathcal{E}_2)&\leq    \protLength/\blockLength \Pr(\mathcal{E}_1)\\
        &\leq
        \sqrt{\protLength}
        \StateOrder\exp\left[ -
    |\usefulset|^{-\coincidingConstant}
    {{\protLength^{1/4}}/{\coincidingConstant}}
    \right]\\
    &=o(1).
\end{align}
It now immediately follows that:
\begin{align}
    \frac{|S_\protLength|}{|\Pi_\protLength|}&=\Pr(S_n)\\
    &=1-\Pr(\mathcal{E}_2)\\
    &=1-o(1)
\end{align}
which concludes the proof. 
\end{proof}


\section{Concluding Remarks\label{section:conclusion}}
In this paper, the problem of simulating an interactive protocol over a pair of binary-input noisy channels is considered. While previous works \cite{schulman1992communication,schulman1996coding,kol2013interactive,InteractiveLowerBound} approach this problem using worst-case assumptions (characterizing the rates in which \textit{all} possible interactive protocols can be simulated), this work restricts the discussion to a specific set of finite-state protocols. A coding scheme is presented that achieves Shannon capacity for all two-state protocol, but 
can not be used to simulate at Shannon capacity for at least one three-state protocol. Then, specific families of finite-state protocols are considered, and Shannon capacity is proved to be achievable for almost all of their members.

Since the proofs in this paper are based on specific coding schemes, proving their failure does not prove the inachievability of Shannon capacity. 
It is also plausible that Shannon capacity is achievable for larger classes of nontrivial interactive protocols using different coding scheme.
A nontrivial upper bound on the ratio between the Shannon capacity and the interactive capacity for a fixed channel (i.e., not in the limit of a very clean channel) still remains an intriguing open question even in the simplest binary symmetric case. 



\appendices
 \section{Proof of Lemma~\ref{lemma:blockwise} \label{appendix:lemma}}
\begin{proof}
	The proof is by straightforward implementation of Gallager's random coding error exponent and the union bound. Due to  \cite{GallagerIT}[Theorem 5.6.4], the probability of decoding error in a single block is upper bounded by:
	\begin{align}
	\Pr(\text{block error})\leq \exp\left(-\frac{b(n)}{R}E_r(R)\right)
	\end{align}
	where $E_r(R)$ (the error exponent) is strictly positive for any $0\leq R<\Cshannon(P_{Y|X})$ and ${b(n)}/{R}$ is the length of the block code. Now, having $l(n)$ independent such blocks, the probability of error in one or more blocks can be upper bounded using the union bound:
	\begin{align}
	&\Pr(\text{error in any block})\leq l(n) \exp\left(-\frac{b(n)}{R}E_r(R)\right)\\
	&=\exp\left(-b(n)\frac{E_r(R)}{R}+\ln l(n)\right)
	\stackrel{(a)}{=}e^{-\Omega(1)}
	=o(1)
	\end{align}
	where $(a)$ is by the assumption that $l(n) = o(e^{b(n)})$.
\end{proof}

\bibliographystyle{IEEEtran}
\bibliography{bibtex_references}

\end{document}